\documentclass[10pt,conference]{IEEEtran}

\usepackage{cite}
\usepackage{amsmath,amssymb,amsfonts}
\usepackage{algorithmic}
\usepackage{graphicx}
\usepackage{textcomp}
\usepackage{amsthm}
\usepackage{epstopdf}
\usepackage{mathtools}
\usepackage[export]{adjustbox}
\usepackage{subcaption}

\newtheorem{theorem}{Theorem}

\newtheorem{proposition}{Proposition}

\newtheorem{corollary}{Corollary}

\newtheorem{remark}{Remark}

\def\nn{\nonumber}

\usepackage{xcolor}
\def\BibTeX{{\rm B\kern-.05em{\sc i\kern-.025em b}\kern-.08em
    T\kern-.1667em\lower.7ex\hbox{E}\kern-.125emX}}
\begin{document}

\title{Information-Theoretic Bounds on the Moments of the Generalization Error of Learning Algorithms}


\author{\IEEEauthorblockN{Gholamali Aminian\textsuperscript{\textsection}, Laura Toni, Miguel R. D. Rodrigues}\\
\IEEEauthorblockA{\textit{ Department of Electronic and Electrical Engineering} \\
\textit{University College London}\\
\{g.aminian, l.toni, m.rodrigues\}@ucl.ac.uk}
}

\maketitle
\begingroup\renewcommand\thefootnote{\textsection}
\footnotetext{The first author is supported by the Royal Society Newton International Fellowship, grant no. NIF\textbackslash R1 \textbackslash 192656 .}
\endgroup

\begin{abstract}
Generalization error bounds are critical to understanding the performance of machine learning models. In this work, building upon a new bound of the expected value of an arbitrary function of the population and empirical risk of a learning algorithm, we offer a more refined analysis of the generalization behaviour of a machine learning models based on a characterization of (bounds) to their generalization error moments. We discuss how the proposed bounds -- which also encompass new bounds to the expected generalization error -- relate to existing bounds in the literature. We also discuss how the proposed generalization error moment bounds can be used to construct new generalization error high-probability bounds.
\end{abstract}
\section{Introduction}\label{Sec:Introduction}
 Machine learning-based approaches are increasingly adopted to solve various prediction problems in a wide range of applications such as computer vision, speech recognition, speech translation, and many more~\cite{shalev2014understanding},~\cite{bengio2017deep}.
 In particular, supervised machine learning approaches learn a predictor -- also known as a hypothesis -- mapping input variables to output variables using some algorithm that leverages a series of input-output examples drawn from some underlying (and unknown) distribution \cite{shalev2014understanding}. It is therefore critical to understand the generalization ability of such a predictor, i.e., how the predictor performance on the training set differs from its performance on a testing set (or on the population).
 
 A recent research direction within the information-theoretic and related communities has concentrated on the development of approaches to characterize the generalization error of \textit{randomized learning algorithms}, i.e. learning algorithms map the set of training examples to the hypothesis according to some probability law {\cite{xu2017information},\cite{raginsky2016information}}.

The characterization of the generalization ability of randomized learning algorithms has come in two broad flavours. One involves determining a bound to the generalization error that holds on average. For example, building upon pioneering work by Russo and Zou~\cite{russo2019much}, Xu and Raginsky~\cite{xu2017information} have derived average generalization error bounds involving the mutual information between the training set and the hypothesis. Bu \textit{et al.}~\cite{bu2020tightening} have derived tighter average generalization error bounds involving the mutual information between each sample in the training set and the hypothesis. Bounds using chaining mutual information have been proposed in \cite{asadi2018chaining}. The upper bounds based on conditional mutual information are proposed in \cite{steinke2020reasoning} and \cite{negrea2019information}. Other authors have also constructed information-theoretic based average generalization error bounds using quantities such as $\alpha$-R\'eyni divergence, $f$-divergence, Jensen-Shannon divergences, Wasserstein distances, or maximal leakage (see~\cite{esposito2019generalization},~\cite{aminian2020jensen},~\cite{lopez2018generalization}, ~\cite{wang2019information},~\cite{rodriguez2021tighter} or~\cite{jiao2017dependence}). 

The other flavour -- known as \textit{probably approximately correct (PAC)-Bayesian} bounds and \textit{single-draw} upper bounds -- involves determining a bound to the generalization error that holds with high probability. The original PAC-Bayesian generalization error bounds have been characterized via a Kullback-Leibler (KL) divergence (a.k.a. relative entropy) between a prior data-free distribution and a posterior data-dependent distribution on the hypothesis space~\cite{mcallester2003pac}. Other slightly different PAC-Bayesian generalization error bounds have also been offered in  ~\cite{thiemann2017strongly},~\cite{catoni2003pac},  \cite{alquier2018simpler}, \cite{ohnishi2021novel} and  \cite{hellstrom2020generalization}. A general PAC-Bayesian framework offering high probability bounds on a convex function of the population risk and empirical risk with respect to a posterior distribution has also been provided in~\cite{germain2015risk}. A PAC-Bayesian upper bound by considering a Gibbs data-dependent prior is provided in \cite{rivasplata2020pac}. Some single-draw upper bounds have been proposed in \cite{xu2017information}, \cite{esposito2019generalization}, and  \cite{hellstrom2020generalization}.


In this paper,  we aspire to offer a more refined analysis of the generalization ability of randomized learning algorithms in view of the fact that the generalization error can be seen as a random variable with distribution that depends on randomized algorithm distribution and the data distribution. 
The analysis of moments of certain quantities arising in statistical learning problems has already been considered in certain works. For example, Russo and and Zou~\cite{russo2019much} have analysed bounds to certain moments of the error arising in data exploration problems, whereas Dhurandhar and Dobra~\cite{dhurandhar2009semi} have analysed bounds to moments of the error arising in model selection problems. Sharper high probability bounds for sums of functions of independent random variables based on their moments, within the context of stable learning algorithms, have also been derived in~\cite{bousquet2020sharper}.
However, to the best of our knowledge, a characterization of bounds to the moments of the generalization error of randomized learning algorithms, allowing us to capture better how the population risk may deviate from the empirical risk, does not appear to have been considered in the literature.

Our contributions are as follows:

\begin{enumerate}

    \item First, we offer a general upper bound on the expected value of a function of the population risk and the empirical risk of a randomized learning algorithm expressed via certain information measures between the training set and the hypothesis.
    
    \item Second, we offer upper bounds on the moments of the generalization error of a randomized learning algorithm deriving from the aforementioned general bound in terms of power information and Chi-square information measures. We also propose another upper bound on the second moment of generalization error in terms of mutual information.
    
    \item Third, we show how to leverage the generalization error moment bounds to construct high-probability bounds showcasing how the population risk deviates from the empirical risk associated with a randomized learning algorithm.

    \item Finally, we show how the proposed results bound the true moments of the generalization error via a simple numerical example.

\end{enumerate}

We adopt the following notation in the sequel. Upper-case letters denote random variables (e.g., $Z$), lower-case letters denote random variable realizations (e.g. $z$), and calligraphic letters denote sets (e.g. $\mathcal{Z}$). The distribution of the random variable $Z$ is denoted by $P_Z$ and the joint distribution of two random variables $(Z_1,Z_2)$ is denoted by $P_{Z_1 Z_2}$. We let $\log (\cdot)$ represent the natural logarithm. We also let $\mathbb{Z}^+$ represent the set of positive integers.

\section{Problem Formulation}
We consider a standard supervised learning setting where we wish to learn a hypothesis given a set of input-output examples; we also then wish to use this hypothesis to predict new outputs given new inputs.

We model the input data (also known as features) using a random variable $X \in \mathcal{X}$ where $\mathcal{X}$ represents the input set; we model the output data (also known as labels) using a random variable $Y \in \mathcal{Y}$ where $\mathcal{Y}$ represents the output set; we also model input-output data pairs using a random variable $Z = (X,Y) \in \mathcal{Z} = \mathcal{X} \times \mathcal{Y}$ where $Z$ is drawn from $\mathcal{Z}$ per some unknown distribution $\mu$. We also let $\mathcal{S} = \{Z_i\}_{i=1}^n$ be a training set consisting of a number of input-output data points drawn i.i.d. from $\mathcal{Z}$ according to $\mathcal{\mu}$.

We represent hypotheses using a random variable $W \in \mathcal{W}$ 
where $\mathcal{W}$ is a hypothesis class. We also represent a randomized learning algorithm via a Markov kernel that maps a given training set $\mathcal{S}$ onto a hypothesis $W$ of the hypothesis class $\mathcal{W}$ according to the probability law $P_{W|S}$.

Let us also introduce a (non-negative) loss function $\ell:\mathcal{W} \times \mathcal{Z}  \rightarrow \mathbb{R}^+$ that measures how well a hypothesis predicts an output given an input. We can now define the population risk and the empirical risk given by:
\begin{align}
&L_P(w,\mu)\triangleq \int_{\mathcal{Z}}\ell(w,z)\mu(z) dz\\
&L_E(w,s)\triangleq\frac{1}{n}\sum_{i=1}^n \ell(w,z_i)
\end{align}
which quantify the performance of a hypothesis $w$ delivered by the randomized learning algorithm on a testing set (population) and the training set, respectively. We can also define the generalization error as follows:
\begin{align} \label{GE}
\text{gen}(P_{W|S},\mu) \triangleq L_P(w,\mu)-L_E(w,s)
\end{align}
which quantifies how much the population risk deviates from the empirical risk. This generalization error is a random variable whose distribution depends on the randomized learning algorithm probabilistic law along with the (unknown) underlying data distribution. Therefore, an exact characterization of the behaviour of the generalization error -- such as its distribution -- is not possible for all learning algorithms.



In order to bypass this challenge, our goal in the sequel will be to derive upper bounds to the moments of the generalization error given by:
\begin{align}\label{Eq: Moments gen}
    \overline{\text{gen}^m}(P_{W|S},\mu)\triangleq\mathbb{E}_{P_{W,S}}[ (\text{gen}(P_{W|S},\mu))^m]
\end{align}
in terms of various divergences and information-theoretic measures.
In particular, we will use the following divergence measures between two distributions $P_X$ and $P_{X^\prime}$ on a common measurable space $\mathcal{X}$:
\begin{itemize}
    \item The KL divergence given by:
    $$D_{KL}(P_X||P_{X^\prime})\triangleq\int_\mathcal{X}P_X(x)\log\left(\frac{P_X(x)}{P_{X^\prime}(x)}\right)dx $$
     \item The $\alpha$-R\'eyni divergence for $1\leq \alpha$ given by~\cite{van2014renyi}:
     $$D_{\alpha}(P_X||P_{X^\prime})\triangleq \frac{1}{\alpha-1}\log \left(\int_\mathcal{X} \left(\frac{P_{X}(x)}{P_{X^\prime}(x)}\right)^\alpha  P_{X^\prime}(x) dx \right)$$
     \item The power divergence of order $t$ given by~\cite{guntuboyina2013sharp}:
    $$D_P^{(t)}(P_X||P_{X^\prime})\triangleq \int_\mathcal{X} \left(\left(\frac{P_{X}(x)}{P_{X^\prime}(x)}\right)^t-1\right) P_{X^\prime}(x) dx$$
    where $D_{\alpha}(P_X||P_{X^\prime})=\frac{\log(D_P^{(t)}(P_X||P_{X^\prime})+1)}{\alpha-1}$ for $\alpha=t$ and $1\leq\alpha$.
    \item The Chi-square divergence given by~\cite{guntuboyina2013sharp}:
    $$\chi^2(P_X||P_{X^\prime})\triangleq \int_\mathcal{X}\frac{(P_X(x)-P_{X^\prime}(x))^2}{P_{X^\prime}(x)}dx$$
    where $\chi^2(P_X||P_{X^\prime})=D_{P}^{(2)}(P_X||P_{X^\prime})$.
\end{itemize}
We also use the following information measures between two random variables $X$ and $X'$ with joint distribution $P_{XX'}$ and marginals $P_X$ and $P_{X'}$:
\begin{itemize}
    \item The mutual information given by:
    $$I(X;X^\prime)\triangleq D_{KL}(P_{X,X^\prime}||P_X\otimes P_{X^\prime})$$
    \item The power information of order $t$ given by:
    $$I_P^{(t)}(X;X^\prime)\triangleq D_P^{(t)}(P_{X,X^\prime}||P_X\otimes P_{X^\prime})$$
    \item The Chi-square information given by:
    $$I_{\chi^2}(X;X^{\prime})\triangleq \chi^2(P_{X,X^\prime}||P_X\otimes P_{X^\prime})$$
    where $I_{P}^{(2)}(X;X^\prime)=I_{\chi^2}(X;X^{\prime})$
\end{itemize}

\section{Bounding Moments of Generalization Error }

We begin by offering a general result inspired from~\cite{alquier2018simpler} bounding the (absolute) expected value of an arbitrary function of the population and empirical risks under a joint measure in terms of the (absolute) expected value of the function of the population and empirical risks under the product measure.

\begin{theorem}\label{Theorem: Power-divergence phi^p}
Consider a measurable function $F(x,y): \mathbb{R}^2\rightarrow\mathbb{R}$. It follows that 
\begin{align}\label{Eq: Power divergence moment upper bound}
    &\left|\mathbb{E}_{P_{W,S}}[F(L_P(W,\mu),L_E(W,S))]\right|\leq\\\nonumber &\quad\mathbb{E}_{P_{W}\otimes P_{S}}[|F(L_P(W,\mu),L_E(W,S))|^q]^{\frac{1}{q}}(I_P^{(t)}(W;S)+1)^{\frac{1}{t}}
\end{align}
where $t,q>1$ such that $\frac{1}{t}+\frac{1}{q}=1$, $P_S=\mu^{\otimes n}$ is the distribution of training set, and $P_W$ and $P_{W,S}$ are the distribution of hypothesis and joint distribution of hypothesis and training set induced by learning algorithm $P_{W|S}$.
\end{theorem}

\begin{proof}
See Appendix \ref{Proof Theorem: Power-divergence phi^p}.
\end{proof}


Theorem \ref{Theorem: Power-divergence phi^p} can now be immediately used to bound the moments of the generalization error of a randomized learning algorithm in terms of a power divergence, under the common assumption that the loss function is $\sigma$-subgaussian.~\footnote{A random variable $X$ is $\sigma$-subgaussian if $E[e^{\lambda(X-E[X])}]\leq e^{\frac{\lambda^2 \sigma^2}{2}}$ for all $\lambda \in \mathbb{R}$.} In the rest of paper, we assume that the loss function $\ell(w,z)$ is $\sigma$- subgaussian under distribution $\mu$ for all $w\in\mathcal{W}$.

\begin{theorem}\label{Theorem: sigma subgassian for moments}
The $m$-th moment of the generalization error of a randomized learning algorithm obeys the bound given by:
\begin{align}\label{Eq: moments subguassian power divergence bound}
 \left| \overline{\text{gen}^m}(P_{W|S},\mu)\right| \leq
  \sigma^m (\frac{mq}{n})^{\frac{m}{2}} e^{m/e}  (I_P^{(t)}(W;S)+1)^{\frac{1}{t}}
\end{align}
provided that $t,q > 1$, $\frac{1}{t}+\frac{1}{q}=1$, $mq>2$ and $mq\in \mathbb{Z}^+$.
\end{theorem}

\begin{proof}
See Appendix \ref{Proof Theorem: sigma subgassian for moments}.
\end{proof}

Theorem \ref{Theorem: sigma subgassian for moments} can also be immediately specialized to bound the moments of the generalization error of a randomized learning algorithm in terms of a chi-square divergence.

\begin{corollary}\label{Theorem: Chi-square moments}
The $m$-th moment of the generalization error of a randomized learning algorithm obeys the bound given by:
\begin{align}\label{Eq: Chi-square Upper bound for Moments}
  \left|\overline{\text{gen}^m}(P_{W|S},\mu)\right|\leq
  \sigma^m (\frac{2m}{n})^{\frac{m}{2}} e^{m/e}  \sqrt{I_{\chi^2}(W;S)+1}
\end{align}
\end{corollary}

\begin{proof}
This corollary follows immediately by setting $t=q=2$ in Theorem \ref{Theorem: sigma subgassian for moments}.
\end{proof}

Interestingly, these moment bounds also appear to lead to a new average generalization error bound complementing existing ones in the literature.

\begin{corollary}\label{Corollary: sigma subgaussian assumption for average}
The average generalization error can be bounded as follows:
\begin{align}
  |\overline{\text{gen}}(P_{W|S},\mu)| \leq
  \sigma \sqrt{\frac{q}{n}} e^{1/e}  (I_P^{(t)}(W;S)+1)^{\frac{1}{t}}
\end{align}
provided that $q \geq 2$ for $q\in \mathbb{Z}^+$.
\end{corollary}

\begin{proof}
This corollary follows immediately by setting $m=1$ in Theorem \ref{Theorem: sigma subgassian for moments}.
\end{proof}

Note that the chi-square information based expected generalization error upper bound is looser than the mutual information based  counterpart in~\cite{xu2017information}.

It is also interesting to reflect about how the generalization error moment bounds decay as a function of the training set size ingested by the learning algorithm. In general, information measures such as power information and chi-square information do not have to be finite, but these information measures can be shown to obey $0 \leq I_{P}^{(t)}(W;S)\leq R^t-1$ and $0 \leq I_{\chi^2}(W;S)\leq R^2-1$, respectively, provided that~\footnote{This condition holds provided that $\mathcal{W}\times \mathcal{S}$ is countable.}
$$\max_{(w,s)\in \mathcal{W}\times\mathcal{S}}\frac{P_{W|S}(w,s)}{P_W(w)}=R < \infty.$$ It follows immediately that the moments of the generalization error are governed by the upper bound given by:
\begin{align}\label{Eq: Chi-square Upper bound for Moments using Bounded}
  \left|\overline{\text{gen}^m}(P_{W|S},\mu)\right|\leq
  \sigma^m (\frac{2m}{n})^{\frac{m}{2}} e^{m/e} R
\end{align}
exhibiting a decay rate of the order $\mathcal{O}(\sqrt{\frac{1}{n^m}})$. Naturally, with the increase in the training set size, one would expect the empirical risk to concentrate around the population risk, and our bounds hint at the speed of such convergence.

It is also interesting to reflect about the tightness of the various generalization error moment bounds. In particular, in view of the fact that it may not be possible to compare directly information measures such as power information and chi-square information, the following Proposition puts forth conditions allowing one to compare the tightness of the bounds portrayed in Theorem \ref{Theorem: sigma subgassian for moments} and Corollary~\ref{Theorem: Chi-square moments} under the condition that the randomized learning algorithm ingests $n$ i.i.d. input-output data examples.

\begin{proposition}\label{Proposition: compare power to chi square}
The power information of order $t$ generalization error $m$-th moment upper bound
\begin{align}
 \left| \overline{\text{gen}^m}(P_{W|S},\mu)\right| \leq
  \sigma^m (\frac{mq}{n})^{\frac{m}{2}} e^{m/e}  (I_P^{(t)}(W;S)+1)^{\frac{1}{t}}
\end{align}
is looser than the chi-sqare information based bound
\begin{align}
  \left|\overline{\text{gen}^m}(P_{W|S},\mu)\right|\leq
  \sigma^m (\frac{2m}{n})^{\frac{m}{2}} e^{m/e}  \sqrt{I_{\chi^2}(W;S)+1}
\end{align}
provided that $(\frac{2(t-1)}{t})^{\frac{mt(t-1)}{(t-2)}}-1 \leq I_P^{(t)}(W;S)$ for $t>2$ with $\frac{mt}{(t-1)}\in \mathbb{Z}^+$.
\end{proposition}

\begin{proof}
See Appendix \ref{Proof Proposition: compare power to chi square}.
\end{proof}

For example, it turns out $1.34^{12}-1\leq I_P^{(3)}(W;S)$ guarantees a chi-square information based generalization error second moment bound to be tighter than the power information of order 3 based bound.

Finally, we offer an additional bound -- applicable only to the second moment of the generalization error -- leveraging an alternative proof route inspired by tools put forth in~\cite[Proposition 2]{russo2019much}.  It does not appear that \cite[Proposition 2]{russo2019much} can be used to generate generalization error higher-order moment bounds in closed form.

\begin{theorem}\label{Theorem: russo approach}
The second moment of the generalization error of a randomized learning algorithm can be bounded as follows:
\begin{equation}\label{Eq: Second moment based on Mutual information}
    \overline{\text{gen}^2}(P_{W|S},\mu)\leq \frac{\sigma^2}{n}\left(16 I(W;S)+9\right)
\end{equation}
\end{theorem}
\begin{proof}
See Appendix \ref{Proof Theorem: russo approach}.
\end{proof}

The next proposition showcases that under certain conditions the mutual information based second moment bound can be tighter than the chi-square information bound.

\begin{proposition}\label{Proposition: 2nd moment chisquare vs mutual information}
 The second moment of generalization error upper based on Chi-square information
\begin{align}
 \overline{\text{gen}^2}(P_{W|S},\mu)\leq
   \frac{\sigma^2}{n} 4 e^{2/e}  \sqrt{I_{\chi^2}(W;S)+1}
\end{align}
is looser than the upper bound based on mutual information in Theorem~\ref{Theorem: russo approach},
\begin{equation}
    \overline{\text{gen}^2}(P_{W|S},\mu)\leq \frac{\sigma^2}{n}\left(16 I(W;S)+9\right)
\end{equation}
provided that $94 \leq I_{\chi^2}(W;S)$. 
\end{proposition}
\begin{proof}
See Appendix \ref{Proof Proposition: 2nd moment chisquare vs mutual information}.
\end{proof}

\section{From Moments to High Probability Bounds}

We now showcase how to use the moment upper bounds to bound the probability that the empirical risk deviates from the population risk by a certain amount, under a \textit{single-draw} scenario where one draws a single hypothesis based on the training data \cite{hellstrom2020generalization}.

Concretely, our following results leverage generalization error moment bounds to construct a generalization error high-probability bound. In particular, we offer a single draw upper bound on the generalization error by leveraging Theorem 2 in conjunction with Markov's inequality, that can be further optimized with respect to the moment order.




\begin{theorem}\label{Theorem: high probabilities by moments}
 It follows that with probability at least $1-\delta$ for some $\delta\in(0,1)$, $t>1$, by considering $\beta=\frac{1}{t-1}\log\left(\frac{I_{P}^{(t)}(W;S)+1}{\delta^t}\right)$ and under distribution $P_{W,S}$ the generalization error obeys:
\begin{align}\label{Eq: single draw by moment bound new}
    &|\text{gen}(P_{W|S},\mu)|\leq \\\nonumber
    &\quad  e^{1/e+1/2}\sqrt{\frac{2t\sigma^2}{n(t-1)}} \sqrt{\log\left(\sqrt[t]{I_{P}^{(t)}(W;S)+1}\right)+\log(\frac{1}{\delta})} 
\end{align}
provided that $2<\beta$ and $\beta\in \mathbb{Z}^+$.
\end{theorem}

\begin{proof}
See Appendix \ref{Proof Theorem: high probabilities by moments}.
\end{proof}
\begin{remark}
The characterization in Theorem~\ref{Theorem: high probabilities by moments} can also be expressed in terms of $\alpha$-R\'eyni divergence, by considering $\beta=\frac{\alpha}{\alpha-1}\log(\frac{1}{\delta})+D_{\alpha}(P_{W,S}||P_W\otimes P_S)$, as follows:
\begin{align}\label{Eq: single draw by moment bound new alpha}
    &|\text{gen}(P_{W|S},\mu)|\leq \\\nonumber
    &\quad  e^{1/e+1/2}\sqrt{\frac{2\sigma^2\left(D_{\alpha}(P_{W,S}||P_W\otimes P_S)+\log(\frac{1}{\delta})\right)}{n}} 
\end{align}
provided that $2<\beta$ and $\beta \in \mathbb{Z}^+$.
\end{remark}
In \cite[Corollary~4]{hellstrom2020generalization}, a single draw upper bound is proposed which depends on two terms of $\alpha$-R\'eyni divergences and the term $\frac{4 \sigma^2}{n}\log(\frac{2}{\delta})$. Our upper bound, \eqref{Eq: single draw by moment bound new alpha}, depends on the  $\alpha$-R\'eyni divergence and also a smaller term  $\frac{2 \sigma^2}{n}\log(\frac{1}{\delta})$.
\begin{corollary}\label{Corollary: chi square result high probability bound}
 It follows that with probability at least $1-\delta$ for some $\delta\in(0,1)$, by considering $\beta=\log\left(\frac{I_{\chi^2}(W;S)+1}{\delta^2}\right)$ and under distribution $P_{W,S}$ the generalization error obeys:
\begin{align}\label{Eq: single draw by moment bound new chi-square}
    &|\text{gen}(P_{W|S},\mu)|\leq \\\nonumber
    &\quad  e^{1/e+1/2}2\sigma \sqrt{\frac{\log\left(\sqrt{I_{\chi^2}(W;S)+1}\right)+\log(\frac{1}{\delta})}{n}} 
\end{align}
provided that $2<\beta$ and $\beta \in \mathbb{Z}^+$.
\end{corollary}
\begin{proof}
This corollary follows immediately by setting $t=2$ in Theorem \ref{Theorem: high probabilities by moments}.
\end{proof}

It is instructive to comment on how this information-theoretic based high-probability generalization error bound compares to other similar information-theoretic bounds such as in \cite{xu2017information}, \cite{esposito2019generalization} and \cite{hellstrom2020generalization}. Our single-draw bound dependence on  $\delta$ (i.e. $\log(\frac{1}{\delta})$) is more beneficial than Xu \textit{et al.}'s bound~\cite[Theorem~3]{xu2017information} dependent on i.e. $(\frac{1}{\delta})$. Our single-draw bound based on chi-square information (along with bounds based on mutual information) is also typically tighter than maximal leakage based single draw bounds \cite{esposito2019generalization}, \cite{hellstrom2020generalization}.



A similar single-draw high probability upper bound based on chi-square information has also been provided in \cite{esposito2019generalization}. The approach pursued to lead to such bound in \cite{esposito2019generalization} is based on $\alpha$-R\'eyni divergence and $\alpha$-mutual information, whereas our approach leading to Corollary~\ref{Corollary: chi square result high probability bound} is based on optimization of bounds to the moments of the generalization error with respect to order of the moments.

\section{Numerical Example}
We now illustrate our generalization error bounds within a very simple setting involving the estimation of the mean of a Gaussian random variable $Z \sim \mathcal{N}(\alpha,\beta^2)$ -- where $\alpha$ corresponds to the (unknown) mean and $\beta^2$ corresponds to the (known) variance -- based on $n$ i.i.d. samples $Z_i$ for $i=1,\cdots,n$.

We consider the hypothesis corresponding to the empirical risk minimizer given by $W=\frac{Z_1+\cdots+Z_n}{n}$. We also consider the loss function given by $$\ell(w,z)=\min((w-z)^2,c^2).$$
In view of the fact that the loss function is bounded within the interval $[0,c^2]$, it is also $\frac{c^2}{2}$-subgaussian so that we can apply the generalization error moments upper bounds offered earlier.

In our simulations, we consider $\alpha=0$, $\beta^2=1$ and $c=\frac{2}{3}$. We compute the true generalization error numerically. 
We also compute chi-square and mutual information bounds to the moments of the generalization error appearing in Corollary~\ref{Theorem: Chi-square moments} and Theorem~\ref{Theorem: russo approach}. We focus exclusively on chi-square information -- corresponding to power information of order 2 -- because it has been established in Proposition~\ref{Proposition: compare power to chi square} that the chi-square information bound can be tighter than the power information one under certain conditions. Both the chi-square information and the mutual information are evaluated numerically. Due to complexity in estimation of chi-square information and mutual information, we consider a relatively small number of training samples.

Fig.\ref{fig:Upper bounds compare moments 1} and Fig.\ref{fig:Upper bounds compare moments 2} demonstrate that the chi-square based bounds to the first and second moment of the generalization error is looser than the mutual information based bounds, as suggested earlier. Fig.\ref{fig:Upper bounds compare moments 3 4} also suggests that higher-order moments (and bounds) to the generalization error decay faster than lower-order ones, as highlighted earlier.

\begin{figure}[t!]
   \includegraphics[width=0.48\textwidth]{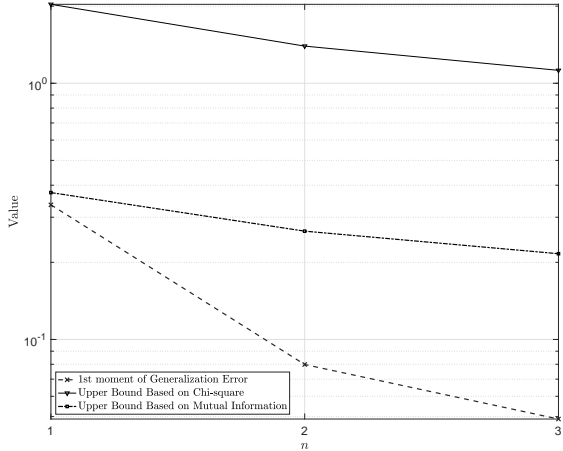}
    \caption{First moment of the generalization error. The figure depicts the true values along with bounds based on mutual information and chi-square information.
    }
    \label{fig:Upper bounds compare moments 1}
\end{figure}
\begin{figure}[t!]
   \includegraphics[width=0.48\textwidth]{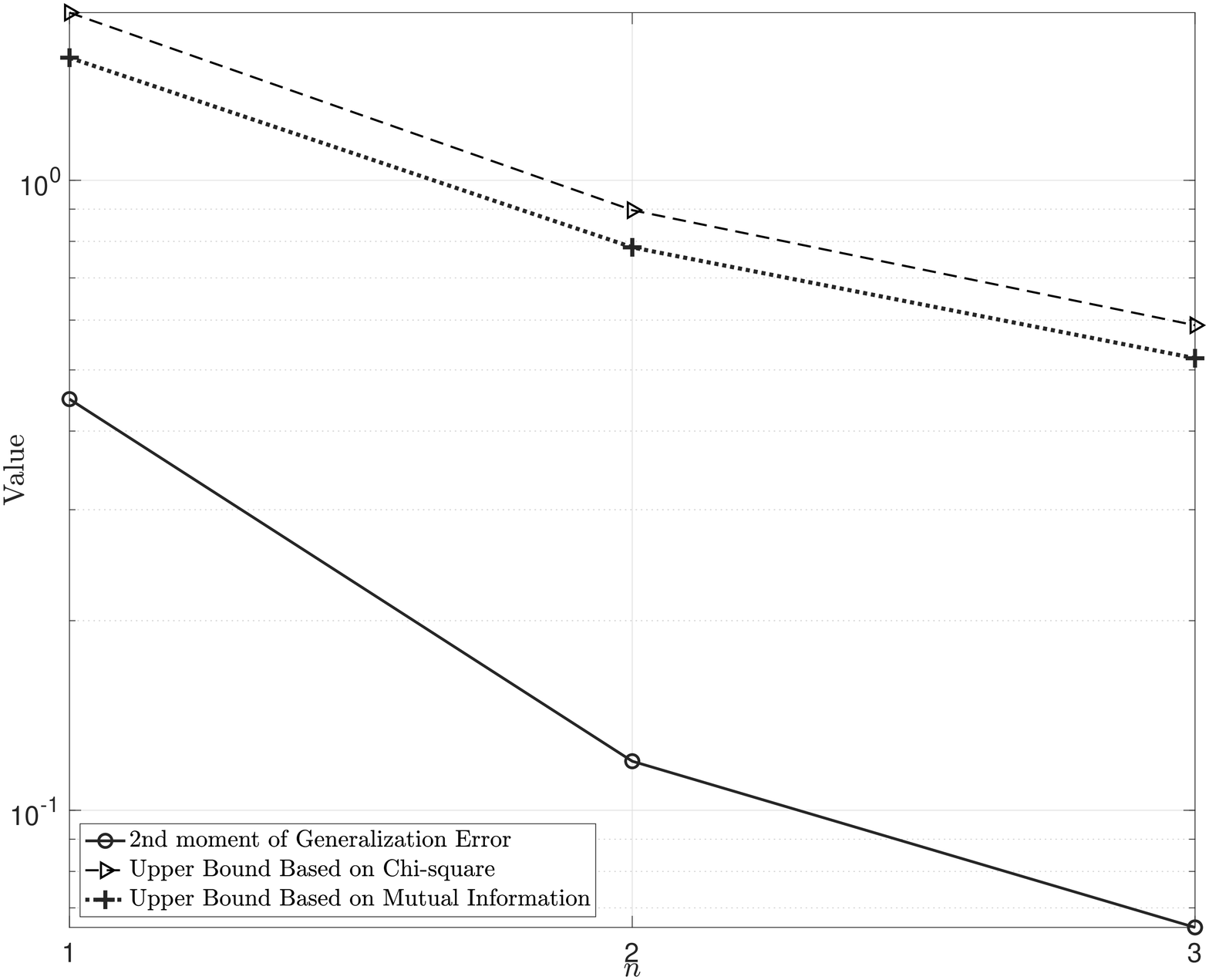}
    \caption{Second moment of the generalization error. The figure depicts the true values along with bounds based on mutual information and chi-square information.
    }
    \label{fig:Upper bounds compare moments 2}
\end{figure}
\begin{figure}[t!]
 \includegraphics[width=0.48\textwidth]{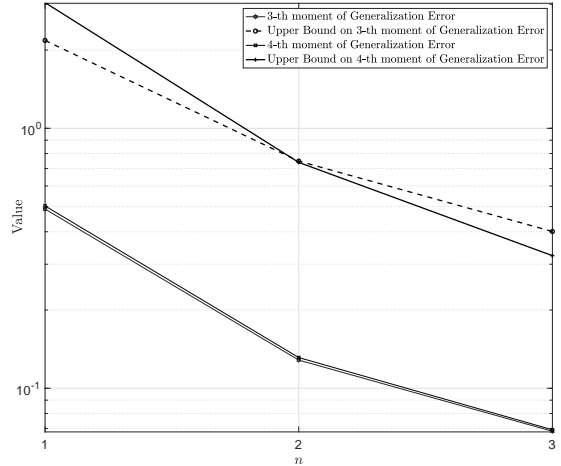}
    \caption{Third and fourth moments of the generalization error. The figure depicts the true values along with bounds based on chi-square information.
    }
    \label{fig:Upper bounds compare moments 3 4}
\end{figure}

\section{Conclusion}\label{Conc}

We have introduced a new approach to obtain information-theoretic oriented bounds to the moments of generalization error associated with randomized supervised learning problems. We have discussed how these bounds relate to existing ones within the literature. Finally, we have also discussed how to leverage the generalization error moment bounds to derive a high probability bounds to the generalization error.

\bibliographystyle{ieeetr}
\bibliography{Refs}

\appendices

\section{Proof of Theorem \ref{Theorem: Power-divergence phi^p}}\label{Proof Theorem: Power-divergence phi^p}
The result follows immediately by noting that:
\begin{align}
    &~~~|\mathbb{E}_{P_{W,S}}[F(L_P(w,\mu),L_E(w,S))]|\\ &\leq\mathbb{E}_{P_{W,S}}[|F(L_P(w,\mu),L_E(w,S))|]\\
    &=\int |F(L_P(w,\mu),L_E(w,S))|\frac{dP_{W,S}}{d(P_W \otimes P_S)}d(P_W \otimes P_S)\\\label{Eq: Holder Inequality}
    &\leq  \mathbb{E}_{P_W \otimes P_S}[|F(L_P(w,\mu),L_E(w,S))|^q]^{\frac{1}{q}}\times \\\nonumber&\quad\left( \int (\frac{dP_{W,S}}{d(P_W \otimes P_S)})^t d(P_W \otimes P_S) \right)^{\frac{1}{t}}\\
    &=\mathbb{E}_{P_W \otimes P_S}[|F(L_P(w,\mu),L_E(w,S))|^q]^{\frac{1}{q}}(I_P^{(t)}(W;S)+1)^{\frac{1}{t}}
\end{align}
where \eqref{Eq: Holder Inequality} is due to H$\Ddot{\text{o}}$lder inequality.
\section{Proof of Theorem \ref{Theorem: sigma subgassian for moments}}\label{Proof Theorem: sigma subgassian for moments}
This result follows from Theorem~\ref{Theorem: Power-divergence phi^p} by considering:
\begin{align}
    F(L_P(w,\mu),L_E(w,s))=(L_P(w,\mu)-L_E(w,s))^m
\end{align}
We now have that:
\begin{align}
    &\left| \overline{\text{gen}^m}(P_{W|S},\mu)\right| \leq\\\nonumber 
  &\mathbb{E}_{P_W}[\mathbb{E}_{P_S}[|(L_P(w,\mu)-L_E(w,s))|^{mq}]]^{\frac{1}{q}} (I_P^{(t)}(W;S)+1)^{\frac{1}{t}}
\end{align}
We also have that:
\begin{align}\label{eq: moment bound sigma sub}
    \mathbb{E}_{P_S}[|(L_P(w,\mu)-L_E(w,s))|^{mq}]^{\frac{1}{q}}\leq \sigma^m e^{m/e}(\frac{mq}{n})^{\frac{m}{2}} 
\end{align}
in view of the fact that (a) the loss function is $\sigma$-subgaussian hence (b) $\text{gen}(P_{W|S},\mu)$ is $\frac{\sigma}{\sqrt{n}}$-subgaussian and (c) \cite[Lemma 1.4]{Philippe-MIT-2015}. In \cite[Lemma 1.4]{Philippe-MIT-2015}, it is assumed that (c) is valid for $mq>2$ and $mq\in \mathbb{Z}^+$.This completes the proof.


\section{Proof of Proposition\ref{Proposition: compare power to chi square}}\label{Proof Proposition: compare power to chi square}

This result follows from the inequality given by~\cite[Corollary 5.6]{guntuboyina2013sharp}:
\begin{align}\label{Eq: Chisquare in compare to power}
    \sqrt{I_{\chi^2}(W;S)+1}\leq (I_P^{(t)}(W;S)+1)^{\frac{1}{2(t-1)}}
\end{align}
holding for $t>2$. We then have that:
\begin{align}
  &\sigma^m (\frac{2m}{n})^{\frac{m}{2}} e^{m/e}  \sqrt{I_{\chi^2}(W;S)+1}\leq\\\nn
  &\sigma^m (\frac{2m}{n})^{\frac{m}{2}} e^{m/e}  (I_P^{(t)}(W;S)+1)^{\frac{1}{2(t-1)}}\leq\\
  &\sigma^m (\frac{mt}{n(t-1)})^{\frac{m}{2}} e^{m/e}  (I_P^{(t)}(W;S)+1)^{\frac{1}{t}}
\end{align}
where the last inequality is valid if $(\frac{2(t-1)}{t})^{\frac{mt(t-1)}{(t-2)}}-1 \leq I_P^{(t)}(W;S)$ for $t>2$ and considering $\frac{mt}{(t-1)}\in \mathbb{Z}^+$.

 \section{Proof of Theorem \ref{Theorem: russo approach}}\label{Proof Theorem: russo approach}
The loss function is assumed to be $\sigma$-subgaussian under distribution $\mu$ for all $w\in \mathcal{W}$ hence -- in view of the fact that the data samples are i.i.d. -- $L_E(W,S)$ is $\frac{\sigma^2}{n}$-subgaussian and also
$\text{gen}(P_{W|S},\mu)$ is $\frac{\sigma^2}{n}$-subgaussian under distribution $P_S$ for all $w\in \mathcal{W}$.


It is possible to establish that the random variable $\text{gen}^2(P_{W|S},\mu)-\mathbb{E}_{P_S}[\text{gen}^2(P_{W|S},\mu)]$ is $(\frac{256 \sigma^4}{n^2},\frac{16 \sigma^2}{n})$-subexponential~\cite[Lemma~1.12]{Philippe-MIT-2015}. \footnote{A random variable $X$ is $(\sigma^2,b)$-subexponential if $\mathbb{E}_{P_X}[e^{\lambda(X-E[X])}]\leq e^{\frac{\lambda^2 \sigma^2}{2}}$ for all $|\lambda|\leq \frac{1}{b}$.} It is worthwhile to mention that the random variable $\text{gen}^2(P_{W|S},\mu)$ is subexponential under distribution $P_S$ for all $w\in\mathcal{W}$. We want to provide the upper bound on the expected value of $\text{gen}^2(P_{W|S},\mu)$ under the joint distribution $P_{W,S}$. Now, we have from the variational representation of the Kullback-Leibler distance that: 
\begin{align}
    &\lambda \mathbb{E}_{P_{W,S}}[ \text{gen}^2(P_{W|S},\mu)]-\log(\mathbb{E}_{P_W\otimes P_S}[e^{\lambda\text{gen}^2(P_{W|S},\mu))}])\\\nonumber 
    &\quad \leq I(W;S)
\end{align}
As the $\text{gen}^2(P_{W|S},\mu)$ is $\frac{16 \sigma^2}{n}$-subexponential under distribution $\mu$ for all $w\in \mathcal{W}$, we have:
\begin{align}
&\log(\mathbb{E}_{P_W\otimes P_S}[e^{\lambda(\text{gen}^2(P_{W|S},\mu)-\mathbb{E}_{P_S}[\text{gen}^2(P_{W|S},\mu)])}])\leq \\\nonumber
&\quad\frac{128\lambda^2\sigma^4}{n^2}
 \quad \text{for}\quad|\lambda|\leq \frac{n}{16\sigma^2}
\end{align}
As $\text{gen}(P_{W|S},\mu)$ is $\frac{\sigma}{\sqrt{n}}$-subgaussian, we also have $\mathbb{E}_{P_S}[\text{gen}^2(P_{W|S},\mu)]\leq \frac{\sigma^2}{n}$ for all $w\in\mathcal{W}$. Therefore the following inequality holds:
\begin{align}
    &\log(\mathbb{E}_{P_W\otimes P_S}[e^{\lambda(\text{gen}^2(P_{W|S},\mu))}])\leq \frac{128\lambda^2\sigma^4}{n^2}+\frac{\lambda\sigma^2}{n}\\\nonumber
&\quad \text{for}\quad|\lambda|\leq \frac{n}{16\sigma^2}
\end{align}
This leads to the inequality:
\begin{align}
    &\mathbb{E}_{P_{W,S}}[ \text{gen}^2(P_{W|S},\mu)]\leq \frac{128\lambda\sigma^4}{n^2}+\frac{\sigma^2}{n}+
     \frac{I(W;S)}{\lambda}
\end{align}
holding for $|\lambda|\leq \frac{n}{16\sigma^2}$.

The final result follows by choosing $\lambda=\frac{n}{16 \sigma^2}$.
\section{Proof of Proposition~\ref{Proposition: 2nd moment chisquare vs mutual information}}\label{Proof Proposition: 2nd moment chisquare vs mutual information}
The result follow from the inequality given by~\cite{dragomir2000some}:
\begin{align}
    I(W;S)\leq \log(I_{\chi^2}(W;S)+1)
\end{align}
We then have that:
\begin{align}
  16I(W;S)+9\leq 16\log(I_{\chi^2}(W;S)+1)+9 
\end{align}
and, for $94 \leq I_{\chi^2}(W;S)$, we also have that:
\begin{align}
  16\log(I_{\chi^2}(W;S)+1)+9 \leq 4e^{2/e}\sqrt{
  I_{\chi^2}(W;S)+1} 
\end{align}

\section{Proof of Theorem~\ref{Theorem: high probabilities by moments}}\label{Proof Theorem: high probabilities by moments}
Consider that:
\begin{align}
    &P_{W,S}(|\text{gen}(P_{W|S},\mu)|\geq r)=\\\nonumber
     &P_{W,S}(|\text{gen}(P_{W|S},\mu)|^{2m}\geq r^{2m})=\\
     &P_{W,S}(\text{gen}(P_{W|S},\mu)^{2m}\geq r^{2m})\leq\\
     &\quad\frac{\mathbb{E}_{P_{W,S}}[\text{gen}(P_{W|S},\mu)^{2m}]}{r^{2m}}\leq\\
     &\quad \sigma^{2m} (\frac{2qm}{n})^{m} e^{2m/e}  \frac{\sqrt[t]{(I_{P}^{(t)}(W;S)+1)}}{r^{2m}}
\end{align}
where the first inequality is due to Markov' inequality and the second inequality is due to Corollary~\ref{Theorem: Chi-square moments}.
Consider also that
$$\delta=\sigma^{2m} \left(\frac{2qm}{n}\right)^{m} e^{2m/e}  \frac{\sqrt[t]{(I_{P}^{(t)}(W;S)+1)}}{r^{2m}} $$
We then have immediately that with probability at least $1-\delta$ under the distribution $P_{W,S}$ it holds that:
\begin{align}\label{Eq: single draw by moment bound}
    &|\text{gen}(P_{W|S},\mu)|\leq \\\nonumber
    &\quad \min_{m>0} \sigma \sqrt{\frac{2mq}{n}} e^{1/e}  \frac{\sqrt[2mt]{I_{P}^{(t)}(W;S)+1}}{\sqrt[2m]{\delta}}
\end{align}
The value of $m$ that optimizes the right hand side in the bound above is given by: $$m^\star=\log\left(\frac{\sqrt[t]{I_{P}^{(t)}(W;S)+1}}{\delta}\right). $$
Based on \eqref{eq: moment bound sigma sub}, the $m^\star$ should also satisfy the conditions, $2<\frac{m^\star t}{t-1}$ and $\frac{m^\star t}{t-1}\in \mathbb{Z}^+$. Therefore, we have $2<\frac{1}{t-1}\log\left(\frac{I_{P}^{(t)}(W;S)+1}{\delta^t}\right)$and $\frac{1}{t-1}\log\left(\frac{I_{P}^{(t)}(W;S)+1}{\delta^t}\right)\in \mathbb{Z}^+$. The result then follows immediately by substituting $m^\star$ in \eqref{Eq: single draw by moment bound}.


\end{document}